\documentclass[letterpaper,onecolumn,prx,aps,showpacs,superscriptaddress,longbibliography,floatfix]{revtex4-2}

\usepackage{amssymb, amsmath, bbm, amsthm, changes, graphicx, dsfont, verbatim, appendix, physics} 

\usepackage[hyperindex, breaklinks]{hyperref}
\hypersetup{
     colorlinks=true,        
     citecolor=blue,    
     filecolor=blue,      		
     urlcolor=blue,           	
    runcolor=cyan,
}

\usepackage{enumitem}   
\usepackage{xcolor}
\usepackage{orcidlink}

\def\Z2{\mathbb Z_2}

\def\b|{\big \|}
\def\dist{\textrm{dist}}
\def\supp{\textrm{supp}}
\def\om{\omega}

\def\l{\lambda}
\def\L{\Lambda}

\theoremstyle{definition}
\newtheorem{Lemma}{Lemma}

\newtheorem{Theorem}[Lemma]{Theorem}
\newtheorem{Corollary}[Lemma]{Corollary}

\begin{document}
\title{Logarithmic lightcones in the multiparticle Anderson model with sparse interactions}

\author{Daniele Toniolo\,\orcidlink{0000-0003-2517-0770}}
\email[]{danielet@alumni.ntnu.no}
\email[]{d.toniolo@ucl.ac.uk}
\affiliation{Department of Physics and Astronomy, University College London, Gower Street, London WC1E 6BT, United
Kingdom}

\affiliation{Department of Mathematics, Eberhard Karls Universit\"{a}t T\"{u}bingen, Auf der Morgenstelle 10, 72076, T\"{u}bingen, Germany}

\author{Sougato Bose\,\orcidlink{0000-0001-8726-0566}}
\affiliation{Department of Physics and Astronomy, University College London, Gower Street, London WC1E 6BT, United
Kingdom}


\begin{abstract}
We prove that the dynamics of the one-dimensional $ XY $ model with random magnetic field perturbed by a sparse set of $ ZZ $ terms with a large coupling constant $ \Delta $ gives rise to Lieb-Robinson (L-R) bounds with a logarithmic lightcone and amplitude proportional to $ \Delta^{-1} $.  These spin systems are equivalent to a set of spinless lattice fermions subjected to a random on site potential and sparse density-density interactions. In the absence of the random magnetic field we also obtain a suppression of the L-R bounds as $ \Delta^{-1} $. To obtain these results we prove a theorem about the effect of compactly supported time-dependent perturbations on the L-R bound of a generic local time-dependent one-dimensional spin systems. Adopting the interaction picture of the dynamics, the large and sparse $ ZZ $ perturbations of the $ XY $ model, with or without disorder, are mapped into high-frequency periodic perturbations. All our results are non-perturbative. 
\end{abstract}

\maketitle

\tableofcontents


\section{Introduction}

Lieb-Robinson (L-R) bounds provide both a qualitative and a quantitative description of the dynamics of a quantum many-body system. In the language of spin systems, given an operator that is acting non trivially on a given bounded set of spins, this region is called its support, the L-R bound gives an upper bound on the spread of its support in time as a consequence of the time evolution given by an underlying local Hamiltonian. The pioneering work of Lieb and Robinson \cite{Lieb_Robinson} has been improved in the last decades especially under the impulse of Hastings \cite{Hastings_2004}. A vast amount of work on this subject has been put forward; recent reviews on the subject are \cite{Nach_2019, Gong_2022, Lucas_review}.

The object of the present work is studying how compactly supported perturbations to a local finite-range one-dimensional Hamiltonian modify the L-R bound of the starting dynamics. We have already provided a result concerning the stability of slow dynamics in our work \cite{Toniolo_2024_1}. Here considering a specific model, the $ XY $ model with and without a random magnetic field, perturbed by a single $ ZZ $ term with a large coupling constant $\Delta$, see equations \eqref{Ham} and \eqref{Ham_no_dis} below, or a sparse set of $ZZ$ terms, we are able to show that the amplitude of the lightcone due to the perturbation becomes proportional to $\Delta^{-1}$. We give an intuitive justification of this result at the beginning of section \ref{applications}, going from the spin to the particle picture using a Jordan-Wigner transformation. We obtain these results as the application of a general theorem on the modifield L-R bound given by a time-dependent perturbations of local Hamiltonian that we prove in \ref{Floquet_Perturbation}. Another work that considers L-R bounds in the context of many-body systems subjected to random fields and possibly interacting is \cite{Nach_2021}.

The one-dimensional $ XY $ model with a random magnetic field is the Fock space analogue of the one-particle Anderson model. A uniform in time L-R bound, as proven in \cite{Sims_Stolz_2012}, shows dynamical Anderson localization. This means that there is no spread of the support of operators in the Heisenberg picture, see equation \eqref{Anderson} for the precise statement. This also implies a bound uniform in time for the dynamically generated entanglement entropy \cite{Nach_2016}. For a general review on the mathematics of the Anderson model see for example \cite{Stolz_review}.

Despite the vast theoretical research on interacting spin models subjected to random fields, the majority of the results is numerical or non rigorous. Recent rigorous results \cite{Elgart_2023,Elgart_2026} show that, at least in the low energy subspace, a logarithmic lightcone appears in the one-dimensional $ XXZ $ model with a random magnetic field. We stress that this is still an upper bound, and up to our knowledge a proof of the fact that the Anderson model is unstable, namely it delocalizes, possibly very slowly, because of interactions, at least in some parameter regimes, is missing. 
A dynamics giving rise to a logarithmic lightcone leads to a logarithmic in time growth of the dynamical entanglement entropy. This has been proven in \cite{Zeng_2023} and in \cite{Toniolo_2024_3} for $\alpha \le 1 $ Rényi entropies. The instability of the Anderson model, from the point of view of the so called orthogonality catastrophe, has been show in \cite{Dietlein_2019}.

Recent results for the evaluation of L-R bounds in random spin systems, especially under the assumption of a heavy tail distribution, are \cite{Baldwin_2023, Baldwin_2025}. We finally like to mention references \cite{De_Roeck_2025} and \cite{Lemm_Lucas_2025}. The first work, following general ideas from \cite{Imbrie_2016}, proves MBL for subintervals of size of order $\log L $ of one-dimensional chains, this suffices to prove subdiffusion for a chain of size $L$. In the work \cite{Lemm_Lucas_2025} starting from a non-resonant conditions the authors show, in any dimension, a non-perturbatively small velocity of ballisitc information transport for all many-body states and the existence of a non-perturbatively long time scale with a logarithmically slow spread of entanglement.

This work is organized as follows: in section \ref{theo} we prove a general theorem on the modification of the L-R bound by perturbations that are compactly supported and possibly time-dependent, this result turns out particularly useful in the  study of high frequency perturbations. In section \ref{applications} we will realized that  a single (or a sparse set of) $ ZZ $ perturbation of the $ XY $ model, with or without random magnetic field, looks like a high frequency perturbation in the interaction picture of the dynamics, so we will apply the result of section \ref{theo} to show the emergence of logarithmic lightcones with suppressed amplitude. We also consider in detail the case of a single $ZZ$ perturbation of the $XY$ model in \ref{XY}.

In one appendix \ref{induction} we will present a particularly concise proof of the L-R  bound for a generic time-dependent nearest neighbour Hamiltonian.


\section{Bounded and time-dependent perturbations of a local dynamics} \label{theo}

\begin{Theorem} \label{Floquet_Perturbation} 
{\it The Hamiltonian $ E(t) $ of a local, one-dimensional spin $ 1/2 $ system, over the lattice $ \{-L,\dots, L\} $ and Hilbert space $ \bigotimes_{j=-L}^{L} \mathds{C}^2 $, in general time-dependent, is assumed to  
give rise to a Lieb-Robinson bound for any pair of operators $ A $ and $ B $, with bounded supports such that $ l:= \dist(\supp \, A,\supp \, B) $
\begin{align} \label{E_L-R}
\b| [T^*_E(t,s) \, A \,T_E(t,s),B] \b| \le     \|A\|  \|B\|   \, f(t,s) \, e^{-\frac{l}{\xi}}
\end{align}
where we have denoted $ f(\cdot,\cdot) $ a positive function, and the time-ordered operator of unitary time evolution 
\begin{align} \label{t_order}
  T_E(t,s) :=  \mathcal{T} \left( e^{ -i \int_s^t du E(u)} \right)
\end{align}
With a time dependent generator $ E(t) $, the time-ordered unitary $ T_E(t_f,t_i) $ depends in general on both the initial time $ t_i $ and the final time $ t_f $ of the evolution. To simplify the notation we set $ f(t):=f(t,0) $, also $ T_H(t):=T_H(t,0) $, for any Hermitian $ H $.

Let us consider an operator $ \l(t) C $, with $ \l(t) $ a real function with primitive function $ \L(t) $, and $ C $ a Hermitian operator with bounded support, such that $ [E(t),C] \neq 0 $ and that at least one among $ [A,C] $ and $ [B,C] $ is vanishing, see figure \ref{supports}. We also assume that for almost every $ s $, the support of $ [E(s),C] $ is independent from $ s $. We denote $ G(t) := E(t) + \l(a t) C $, $ a \ge 0 $, the total Hamiltonian and $ T_G(t) $ the corresponding operator of unitary evolution. We now present two independent L-R bounds \eqref{direct} and \eqref{by_part}, what's the best depends on the specific form of $\l(at) $ and the value of $a$. 

It holds:
\begin{align} \label{direct}
 \b| [T^*_G(t) \, A \,T_G(t),B] \b|  \le \|A\|  \|B\| \, f(t) \, e^{-\frac{l}{\xi}} + 2 \|A\|  \|B\| \|C\| \bar{f}(t)\, e^{-\frac{r}{\xi}}
\end{align}
with 
\begin{align}
 \bar{f}(t) := \max \big \{\int_0^t ds |\l(as)| f(s,0) \,, \int_0^t ds |\l(as)| f(t,s) \big \}
\end{align}
and
\begin{align} \label{r_def}
 r:= \max \{\dist(\supp \, A,\supp \, C),\dist(\supp \, B,\supp \, C)\}
\end{align}

It also holds:
\begin{align} \label{by_part}
\b| [T^*_G(t) \, A \,T_G(t),B] \b| \le   \|A\|  \|B\|   \, f(t) e^{-\frac{l}{\xi}}
 + \sup_{s \in [0,t]} \frac{|\L(as)|}{a} 2 \|A\|  \|B\| \tilde{f}_C(t) e^{-\frac{d}{\xi}} 
\end{align}
with 
\begin{align}\label{I_L-R}
 \tilde{f}_C(t) :=  \|C\|  f(t) + \max \Big \{
  \int_0^t ds f(s,0) \left( 2 |\l(as)| \|C\|^2   + \|[E(s),C]\| \right) \, , \, \int_0^t ds f(t,s) \left( 2 |\l(as)| \|C\|^2   + \|[E(s),C]\| \right) \Big \}
\end{align}
and
\begin{align} \label{d_def}
 d:= \max \{\dist(\supp \, A,\supp \, [E(s),C]),\dist(\supp \, B,\supp \, [E(s),C])\}
\end{align}
}

\end{Theorem}

{\it Remarks}. We have provided two L-R bounds for the same dynamics. The bound in equation \eqref{direct} is the one for $ a=0 $ (meaning that the perturbation $ \lambda(0) \, C $ is time-independent) or $ a \ll 1 $, instead \eqref{by_part} ideally applies when $ a \gg 1 $.

For the intermediate values of $ a $, say $ a = O(1) $ we cannot identify in general which one among the bounds \eqref{direct} and \eqref{by_part} would be the better one, therefore this has to be identified case by case given the specific Hamiltonian theorem \ref{Floquet_Perturbation} is applied to. Applications of theorem \ref{Floquet_Perturbation} to the Anderson model (meaning the $XY$ model with a random magnetic field), and the $ XY $ model perturbed by a single $ ZZ $ term are provided in lemma \ref{single_ZZ} (see also equations \eqref{L-R} and \eqref{delta_values} for the informal version of it), and corollary \ref{single_ZZ_no_dis}.

The assumption about the relative positions of the supports of $ A $, $ B $ and $ C $, simply means that the support of $ C $ does not ``connect'' the supports of $ A $ and $ B $, but it can overlap with one of them, as schematically represented in figure \ref{supports}.

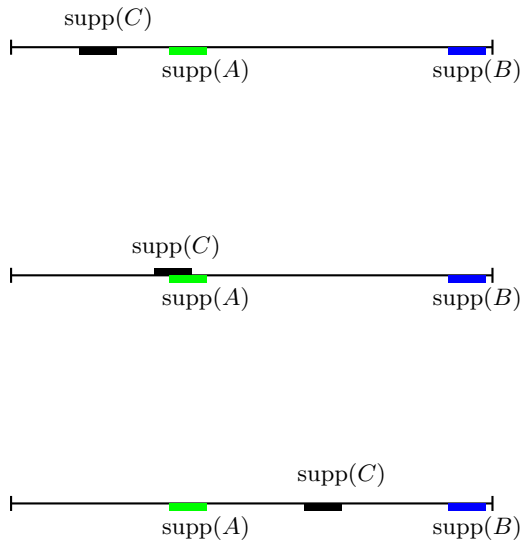
\begin{figure}[h]
\setlength{\unitlength}{1mm}


\setlength{\unitlength}{1mm} 
\begin{picture}(140,20)(10,0)

\thicklines

\put(41,17){\line(0,1){2}}
\put(105,17){\line(0,1){2}}

\put(105,18){\line(-1,0){64}}

\linethickness{1mm}
\put(55,17.5){\line(-1,0){5}}

{\color{green}\put(67,17.5){\line(-1,0){5}}}

{\color{blue}\put(103,17.5){\line(-1,0){5}}}

\thinlines

\put(47,21){$ \supp(C) $}
\put(60,14){$ \supp(A) $}
\put(96,14){$ \supp(B) $}



\end{picture}


\setlength{\unitlength}{1mm} 
\begin{picture}(140,20)(10,0)

\thicklines

\put(41,17){\line(0,1){2}}
\put(105,17){\line(0,1){2}}

\put(105,18){\line(-1,0){64}}

\linethickness{1mm}
\put(65,18.5){\line(-1,0){5}}

{\color{green}\put(67,17.5){\line(-1,0){5}}}

{\color{blue}\put(103,17.5){\line(-1,0){5}}}

\thinlines

\put(56,21){$ \supp(C) $}
\put(60,14){$ \supp(A) $}
\put(96,14){$ \supp(B) $}



\end{picture}


\setlength{\unitlength}{1mm} 
\begin{picture}(140,20)(10,0)

\thicklines

\put(41,17){\line(0,1){2}}
\put(105,17){\line(0,1){2}}

\put(105,18){\line(-1,0){64}}

\linethickness{1mm}
\put(85,17.5){\line(-1,0){5}}

{\color{green}\put(67,17.5){\line(-1,0){5}}}

{\color{blue}\put(103,17.5){\line(-1,0){5}}}

\thinlines

\put(78,21){$ \supp(C) $}
\put(60,14){$ \supp(A) $}
\put(96,14){$ \supp(B) $}



\end{picture}

\caption{Sketching a possible configuration of supports of the operators $ A $, $ B $ and the perturbation $ \lambda \, C $, to illustrate the condition ``at least one among $ [A,C] $ and $ [B,C] $ is vanishing''. } 
\label{supports}
\end{figure}

The assumption about the constancy in time (for almost every $ t $) of the support of $ [E(t), C] $ aligns with the most usual models in physics where the time dependence of the Hamiltonian is given by time-dependent functions, for example $ \cos(\nu t) $ or $ \sin (\nu t) $ for Floquet models, multiplying local terms.

Assuming for example that $ E(s) $ is a sum of nearest neighbour terms and that the size of the support of $ C $ is two, then $ d \ge \frac{l-3}{2} $, therefore: $ d \gtrsim \frac{l}{2} $.

\begin{proof}
 We start the proof with a note on the interaction picture. Denoting  $ T_\alpha(t) $ and $ T_\beta(t) $ the time ordered unitary operators with generators given by two generic time dependent Hermitian matrices $ \alpha(t) $ and $ \beta(t) $, it holds:
 \begin{align} \label{time_dep_int_pic}
  T_\beta(t) = T_\alpha(t) \mathcal{T} \left[ e^{-i \int_0^t ds T^*_\alpha(s)(\beta(s)-\alpha(s)) T_\alpha(s)} \right]
 \end{align}
in fact:
\begin{align}
& i \frac{d}{dt} \left( T_\alpha(t) \mathcal{T} \left[ e^{-i \int_0^t ds T^*_\alpha(s)(\beta(s)-\alpha(s)) T_\alpha(s)} \right] \nonumber \right) \\ 
& = \alpha(t) T_\alpha(t) \mathcal{T} \left[ e^{-i \int_0^t ds T^*_\alpha(s)(\beta(s)-\alpha(s)) T_\alpha(s)} \right] +  T_\alpha(t) \left[ T^*_\alpha(t)( \beta(t) - \alpha(t) ) T_\alpha(t) \right] T \left[ e^{-i \int_0^t ds T^*_\alpha(s)(\beta(s) - \alpha(s)) T_\alpha(s)} \right] \\
 & = \alpha(t) T_\alpha(t) \mathcal{T} \left[ e^{-i \int_0^t ds T^*_\alpha(s)(\beta(s)-\alpha(s)) T_\alpha(s)} \right] + (\beta(t) - \alpha(t)) T_\alpha(t) \mathcal{T} \left[ e^{-i \int_0^t ds T^*_\alpha(s)(\beta(s)-\alpha(s)) T_\alpha(s)} \right] \\
 & = \beta(t) T_\alpha(t) \mathcal{T} \left[ e^{-i \int_0^t ds T^*_\alpha(s)(\beta(s)-\alpha(s)) T_\alpha(s)} \right]
\end{align}
The equality \eqref{time_dep_int_pic} then holds since both sides are the unique solution of the differential equation $ i \frac{d}{dt} T_\beta(t) =  \beta(t)  T_\beta(t) $, with initial condition $ T_\beta(0) = \mathds{1} $.

We use \eqref{time_dep_int_pic} setting $ \beta(t) = G(t) $, with $ G(t) $ given in the statement of theorem \ref{Floquet_Perturbation}, and $ \alpha(t) = E(t) $, and then define
\begin{align}
 I(t) := T^*_E(t)(G(t)-E(t))T_E(t) = T^*_E(t) \, \l(at) C \,  T_E(t) 
\end{align}
The distance $ r $, as defined in \eqref{r_def}, that enters in the RHS of \eqref{direct}, emerges by two possible approaches: in one case we ``bring down'' $ I $ on the first entry of the commutator on the LHS of  \eqref{direct}, in the other case we bring it down on the second entry of the commutator. Let us explain the details of this idea.

We assume $ [ C , B ] = 0 $, this corresponds to considering the case where $ C $ is supported far away from $ B $.

\begin{align} \label{split}
 \b| [T^*_G(t) \, A \,T_G(t),B] \b| = \b| [T^*_I(t) T^*_E(t) \, A \,T_E(t) T_I(t),B] \b| = \b| [ T^*_E(t) \, A \,T_E(t), T_I(t) B T^*_I(t)] \b|
\end{align}
Let us consider a generic initial instant $ t_i = s $ for the time evolution, then using:
\begin{align} \label{54}
& T_I(t):= T_I(t,0)=T_I(t,s) T_I(s,0) \Rightarrow T_I(t,s) = T_I(t,0) T^*_I(s,0) \\
& \Rightarrow i \frac{d}{ds} T_I(t,s) = T_I(t,0) i \frac{d}{ds} T^*_I(s,0) = -T_I(t,0)T^*_I(s,0)I(s)  =  - T_I(t,s) I(s)
\end{align}
we have:
\begin{align} \label{der_second_arg}
 T_I(t) B T^*_I(t) - B & = -\int_0^t ds \, \frac{d}{ds} \left(T_I(t,s) B T^*_I(t,s) \right) \\ & = -i \int_0^t ds \, \left( T_I(t,s) I(s) B T^*_I(t,s) - T_I(t,s) B I(s) T^*_I(t,s) \right) \\
 & = - i \int_0^t ds \, T_I(t,s)[I(s),B]T^*_I(t,s) \\  \label{gen_t}
 & = - i \int_0^t ds \, \l(as) T_I(t,s) [T^*_E(s) C T_E(s), B ] T^*_I(t,s) 
\end{align}

It follows that:
\begin{align} 
 \b| [T^*_G(t) \, A \,T_G(t),B] \b| & \le \b| [ T^*_E(t) \, A \,T_E(t), B] \b| + 2 \|A\|  \int_0^t ds \b|[I(s),B] \b| \\
 & \le \|A\|  \|B\| \, f(t) \, e^{-\frac{l}{\xi}} + 2 \|A\|  \|B\| \|C\| \int_0^t ds |\l(as)|f(s,0) \, e^{-\frac{\dist(\supp \, B,\supp \, C)}{\xi}} \label{direct_B}
\end{align}
Equation \eqref{direct_B} proves bound \eqref{direct} when the $\max$ in the definition of $\bar{f} $ is given by the first entry.

Let us now perform an integration by part on the RHS of equation \eqref{gen_t}.
\begin{align} 
 & \int_0^t ds \l(as) T_I(t,s) [T^*_E(s) C T_E(s), B ] T^*_I(t,s) \label{first}  \\
 & = \frac{\L(as)}{a}  T_I(t,s) [T^*_E(s) C T_E(s), B ] T^*_I(t,s) |_{s=0}^{s=t} -  \int_0^t ds \frac{\L(as)}{a} \frac{d}{ds} \big( T_I(t,s) [T^*_E(s) C T_E(s), B ] T^*_I(t,s) \big) \label{23} \\
 & =  \frac{\L(at)}{a}  [T^*_E(t) C T_E(t), B ] \nonumber \\
 & - i \int_0^t ds \frac{\L(as)}{a}  \Big( T_I(t,s) \big[ I(s), [T^*_E(s) C T_E(s), B ]\big]T^*_I(t,s) + T_I(t,s) \big[T^*_E(s) [E(s), C] T_E(s), B \big]T^*_I(t,s)  \Big) 
\end{align}
Taking the norm of \eqref{first} we get:
\begin{align} 
 & \| \int_0^t ds \l(as) T_I(t,s) [T^*_E(s) C T_E(s), B ] T^*_I(t,s)  \| \label{27}  \\
 & \le  \frac{|\L(at)|}{a} \| [T^*_E(t) C T_E(t), B ] \| + \int_0^t ds \frac{|\L(as)|}{a}  \Big( 2 \| I(s) \| \| [T^*_E(s) C T_E(s), B ] \| + \| \bigskip[T^*_E(s) [E(s), C] T_E(s), B \big] \|  \Big)  \\
 & \le \sup_{s \in [0,t]} \frac{|\L(as)|}{a} \|C\| \|B\| f(t) e^{-\frac{\dist(\supp \, B,\supp \, C)}{\xi}} \nonumber \\
 &  + \sup_{s \in [0,t]}  \frac{|\L(as) |}{a} \|B\|  \int_0^t ds f(s) \left( 2 |\l(as)| \|C\|^2  e^{-\frac{\dist(\supp \, B,\supp \, C)}{\xi}} + \|[E(s),C]\| e^{-\frac{\dist(\supp \, B,\supp [E(s),C])}{\xi}}  \right)  
\end{align}

Replacing \eqref{der_second_arg} into \eqref{split} we obtain, under the assumption $ [ C , B ] = 0 $:
\begin{align} \label{up_B}
 & \b| [T^*_G(t) \, A \,T_G(t),B] \b|  \le \|A\|  \|B\|   \, f(t) e^{-\frac{l}{\xi}} + \sup_{s \in [0,t]} \frac{|\L(as)|}{a} 2 \|A\| \|C\| \|B\| f(t) e^{-\frac{\dist(\supp B,\supp C)}{\xi}} \nonumber \\
 &  + \sup_{s \in [0,t]}  \frac{|\L(as)|}{a} 2 \|A\| \|B\|  \int_0^t ds f(s) \left( 2 |\l(as)| \|C\|^2  e^{-\frac{\dist(\supp B,\supp C)}{\xi}} + \|[E(s),C]\| e^{-\frac{\dist(\supp B,\supp [E(s),C])}{\xi}}  \right) \\
& \le \|A\|  \|B\|   \, f(t) e^{-\frac{l}{\xi}} \nonumber \\
& + \sup_{s \in [0,t]} \frac{|\L(as)|}{a} 2 \|A\|  \|B\| \Big( \|C\|  f(t) +  \int_0^t ds f(s) \left( 2 |\l(as)| \|C\|^2   + \|[E(s),C]\| \right)  \Big) e^{-\frac{\dist(\supp B,\supp [E(s),C])}{\xi}}  \label{49}
\end{align}
In equation \eqref{49} we have used that under the assumption $ [E(t),C] \neq 0 $, it is $ \dist(\supp \, B,\supp [E(s),C]) \le \dist(\supp \, B,\supp \, C) $. We have also used the assumption that for almost every $ s $, the $\supp \, [E(s),C] $ is independent from $ s $. This is an assumption to simplify the statement of the theorem, that also fits with the applications that we are going to present in section \ref{applications}.

Let us know assume instead that  $ [ C , A ] = 0 $.
In 
\begin{align} \label{on__A}
 \b| [T^*_G(t) \, A \,T_G(t),B] \b| = \b| [T^*_I(t) T^*_E(t) \, A \,T_E(t) T_I(t),B] \b|
\end{align}
we use:
\begin{align} \label{on_A}
& T^*_I(t) T^*_E(t) \, A \,T_E(t) T_I(t) - T^*_E(t) \, A \,T_E(t)  = \int_0^t ds \frac{d}{ds} \left(T^*_I(s) T^*_E(t) \, A \,T_E(t) T_I(s) \right) \\
& = i \int_0^t ds T^*_I(s) [I(s), T^*_E(t) \, A \,T_E(t)] T_I(s) = i \int_0^t ds \l(as) T^*_I(s) [T^*_E(s)CT_E(s), T^*_E(t) \, A \,T_E(t)] T_I(s)  \label{before_part}
\end{align}
We have that:
\begin{align} 
 \b| [T^*_G(t) \, A \,T_G(t),B] \b| \le & \b| [ T^*_E(t) \, A \,T_E(t), B] \b| + 2 \|B\|  \int_0^t ds |\l(as)| \b| [T^*_E(s)CT_E(s), T^*_E(t) \, A \,T_E(t)] \b|
\end{align}
With $ \| [T^*_E(s)CT_E(s), T^*_E(t) \, A \,T_E(t)] \| = \| [C, T_E(s) T^*_E(t) \, A \,T_E(t)T^*_E(s)] \| = \| [C, T^*_E(t,s) \, A \,T_E(t,s)] \| $, it follows:
\begin{align} \label{direct_A}
 \b| [T^*_G(t) \, A \,T_G(t),B] \b| \le \|A\|  \|B\| \, f(t) \, e^{-\frac{l}{\xi}} + 2 \|A\|  \|B\| \|C\| \int_0^t ds |\l(as)|f(t,s) \, e^{-\frac{\dist(\supp \, A,\supp \, C)}{\xi}}
\end{align}
Equation \eqref{direct_B} together with \eqref{direct_A} prove the upper bound \eqref{direct}.

Continuing from the RHS of \eqref{before_part} and performing a by part integration 
\begin{align}
& T^*_I(t) T^*_E(t) \, A \,T_E(t) T_I(t) - T^*_E(t) \, A \,T_E(t)  = i \frac{\L(as)}{a} T^*_I(s) [T^*_E(s)CT_E(s), T^*_E(t) \, A \,T_E(t)] T_I(s) |_{s=0}^{s=t} \nonumber \\
& - i \int_0^t ds \frac{\L(as)}{a} \frac{d}{ds} \Big( T^*_I(s) [T^*_E(s)CT_E(s), T^*_E(t) \, A \,T_E(t)] T_I(s) \Big) \label{36} \\
& = - i \frac{\L(0)}{a} [C, T^*_E(t) \, A \,T_E(t)] \nonumber \\
& + \int_0^t ds \frac{\L(as)}{a} T^*_I(s) \Big( \big[ I(s), [T^*_E(s)CT_E(s), T^*_E(t) \, A \,T_E(t)] \big] + \big[T^*_E(s)[E(s),C]T_E(s), T^*_E(t) \, A \,T_E(t) \big]  \Big) T_I(s) \label{53}
\end{align}

The norm of \eqref{53} is upper bounded as:
\begin{align} \label{54}
&  \frac{|\L(0)|}{a} \| [C, T^*_E(t) \, A \,T_E(t)] \| \nonumber \\
& + \sup_{s\in [0,t]} \frac{|\L(as)|}{a}  \int_0^t ds \Big( 2 |\l(as)| \|C\| \| [T^*_E(s)CT_E(s), T^*_E(t) \, A \,T_E(t)] \| + \| \big[T^*_E(s)[E(s),C]T_E(s), T^*_E(t) \, A \,T_E(t) \big] \| \Big)  
\end{align}

Overall the norm of \eqref{on__A} is upper bounded, with the assumption $ [ C , A ] = 0 $, as:
\begin{align} 
 & \b| [T^*_G(t) \, A \,T_G(t),B] \b| \le  \|A\| \|B\| f(t) e^{-\frac{l}{\xi}} + \frac{|\L(0)|}{a} 2 \|B\| \|A\| \|C\| f(t) e^{-\frac{\dist(\supp A,\supp C)}{\xi}} \nonumber \\
 & + \sup_{s\in [0,t]} \frac{|\L(as)|}{a} 2 \|B\| \|A\| \int_0^t ds f(t,s)  \Big( 2 |\l(as)|  \|C\|^2 e^{-\frac{\dist(\supp A,\supp C)}{\xi}}  +  \|[E(s),C]\| e^{-\frac{\dist(\supp A,\supp [E(s),C])}{\xi}} \Big) \\
 & \le \|A\|  \|B\|   \, f(t) e^{-\frac{l}{\xi}} \nonumber \\
& + \sup_{s \in [0,t]} \frac{|\L(as)|}{a} 2 \|A\|  \|B\| \Big( \|C\|  f(t) +  \int_0^t ds f(t,s) \left( 2 |\l(as)| \|C\|^2   + \|[E(s),C]\| \right)  \Big) e^{-\frac{\dist(\supp A,\supp [E(s),C])}{\xi}} \label{up_A}
\end{align}

We define:
\begin{align}
 \tilde{f}_C(t) := \max\{\|C\|  f(t) +  \int_0^t ds f(s) \left( 2 |\l(as)| \|C\|^2   + \|[E(s),C]\| \right),\|C\|  f(t) +  \int_0^t ds f(t,s) \left( 2 |\l(as)| \|C\|^2   + \|[E(s),C]\| \right)\}
\end{align}

At this point we have two upper bounds, equations \eqref{49} and \eqref{up_A}, for the same quantity $ \b| [T^*_G(t) \, A \,T_G(t),B] \b| $, we take as the final upper bound the smaller among the two. This is essentially determined by the relative position of $ C $ with respect to $ A $ and $ B $. Then the final bound, as in equation \eqref{by_part} is:
\begin{align} \label{final}
  \b| [T^*_G(t) \, A \,T_G(t),B] \b| \le   \|A\|  \|B\|   \, f(t) e^{-\frac{l}{\xi}}
 + \sup_{s \in [0,t]} \frac{|\L(as)|}{a} 2 \|A\|  \|B\| \tilde{f}_C(t) e^{-\frac{d}{\xi}} 
\end{align}
with 
\begin{align}
 d:= \max \{\dist(\supp A,\supp [E(s),C]),\dist(\supp B,\supp [E(s),C])\}
\end{align}

\end{proof}

\subsection{Further integrations by part do not improve the bound} \label{higher_orders}

Performing further integrations by part, after the ones we have made in \eqref{23} and \eqref{36}, do not improve the leading order value of the L-R in \eqref{by_part}. The reason is that, assuming on one hand that a primitive function of $ \Lambda(t) $ exists, let us call it $ {\tilde \Lambda} (t) $, that would bring a factor $ \frac{1}{a^2} $, on the other hand a derivative of $ I(t):=  T^*_E(t) \, \l(at) C \,  T_E(t) $ gives a factor $ a $, therefore the leading order contribution of a further integration by part would be $ \approx \frac{1}{a} $ bringing no improvement. The same holds true for further integrations by part. Let us show this explicitly. Returning to equation \eqref{on__A}
 
\begin{align}
&[T^*_G(t) \, A \,T_G(t),B] = [T^*_I(t) T^*_E(t) \, A \,T_E(t) T_I(t),B] \\
& = [T^*_E(t) \, A \,T_E(t),B] + \Big[ \int_0^t ds \frac{d}{ds} \Big( T^*_I(s) T^*_E(t) \, A \,T_E(t) T_I(s) \Big),B \Big] \\
& = [T^*_E(t) \, A \,T_E(t),B] + \Big[ i \int_0^t ds T^*_I(s) [I(s), T^*_E(t) \, A \,T_E(t)] T_I(s) ,B \Big] \\
& = [T^*_E(t) \, A \,T_E(t),B] + \Big[ i \int_0^t ds \l(as) T^*_I(s) [T^*_E(s)CT_E(s), T^*_E(t) \, A \,T_E(t)] T_I(s) ,B \Big] \label{before_int_part}
\end{align}
Performing a by part integration in \eqref{before_int_part} as done in \eqref{36} and continuing from there:
\begin{align}
& i \frac{\L(as)}{a} T^*_I(s) [T^*_E(s)CT_E(s), T^*_E(t) \, A \,T_E(t)] T_I(s) |_{s=0}^{s=t}  - i \int_0^t ds \frac{\L(as)}{a} \frac{d}{ds} \Big( T^*_I(s) [T^*_E(s)CT_E(s), T^*_E(t) \, A \,T_E(t)] T_I(s) \Big) \\
& = - i \frac{\L(0)}{a} [C, T^*_E(t) \, A \,T_E(t)]  \nonumber \\
& + \int_0^t ds \frac{\L(as)}{a} T^*_I(s) \Big( \big[ I(s), [T^*_E(s)CT_E(s), T^*_E(t) \, A \,T_E(t)] \big] + \big[T^*_E(s)[E(s),C]T_E(s), T^*_E(t) \, A \,T_E(t) \big]  \Big) T_I(s) \label{88}
\end{align}
Performing a further integration by part on the first integral in \eqref{88} we get:
\begin{align}
& \int_0^t ds \frac{\L(as)}{a} T^*_I(s)  \big[ I(s), [T^*_E(s)CT_E(s), T^*_E(t) \, A \,T_E(t)] \big] T_I(s)  \\
& = \frac{{\tilde \L}(as)}{a^2} T^*_I(s)  \big[ I(s), [T^*_E(s)CT_E(s), T^*_E(t) \, A \,T_E(t)] \big] T_I(s) |_{s=0}^{s=t} \nonumber \\ 
& - \int_0^t ds \frac{{\tilde \L}(as)}{a^2} \frac{d}{ds} \Big( T^*_I(s)  \big[ I(s), [T^*_E(s)CT_E(s), T^*_E(t) \, A \,T_E(t)] \big] T_I(s) \Big)
\end{align}
As anticipated the derivative of $ I(s):=  T^*_E(s) \, \l(as) C \,  T_E(s) $ gives rise to a factor $ a $, therefore the order of this term is $ \frac{1}{a} $, that is the leading one in $ a $. This is the same order that was already emerging after a single integration by part as in \eqref{by_part}.


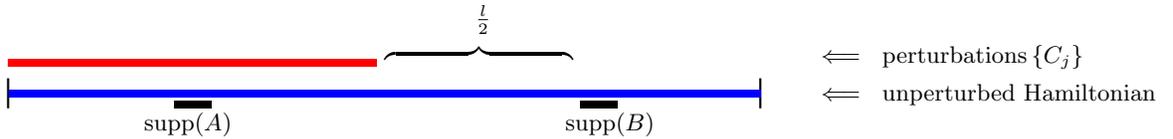
\begin{figure}[h!]
\setlength{\unitlength}{1mm} 

\begin{picture}(180,45)(-30,-40)

\thicklines

\put(-15,-24){\line(0,1){4}}
\put(85,-24){\line(0,1){4}}
\put(-15,-22){\line(1,0){100}}

\linethickness{1mm}

\put(12,-23.5){\line(-1,0){5}}

{\color{red} \put(-15,-18){\line(1,0){49}}}

{\color{blue} \put(-16,-22){\line(1,0){100}}}
\put(65,-23.5){\line(-1,0){5}}

\thinlines

\put(2,-27){$ \supp(A) $}
\put(58,-27){$ \supp(B) $}
\put(34,-18){$\overbrace{\hspace{25mm}}$}
\put(46,-13){$ \frac{l}{2}$}

\put(10,-29){$\underbrace{\hspace{50mm}}$}
\put(35,-34){$ l$}

\put(92,-23){$ \Longleftarrow \hspace{2mm}  \textrm{unperturbed Hamiltonian}$}

\put(92,-18){$ \Longleftarrow \hspace{2mm} \textrm{perturbations} \, \{ C_j \} $}

\end{picture}
\caption{A set of perturbations $ \{ C_j \} $ that cover the support of $ A $ and extends from the left end of the system to roughly half the way in between $ A $ and $ B $. } 
\label{covering}
\end{figure}

\subsection{From one to many perturbations} \label{many_per}

The theory that we have developed in theorem \ref{Floquet_Perturbation} can be easily extended to the case of sparse perturbations but with the important difference that if for one perturbation $ \lambda(at) C $ in \ref{Floquet_Perturbation} we were able to describe the resulting L-R bounds \eqref{direct} and \eqref{by_part} for every possible positions of the supports of $ A $ and $ B $ relative to $ C $ (up to the trivial cases), when more perturbations are taken into account then the bound will be valid only either for a fixed pair of operators $ A $ and $ B $ and a certain distribution of perturbations with respect to that pair, or conversely fixing the set of perturbations and considering certain pairs of operators $ A $ and $ B $.

We explicitly consider two configurations of perturbations sketched in figures \ref{covering} and \ref{fig_left_right}. We consider the perturbation $ \sum_j \lambda_j(a_j s) C_j $ ($ j $ runs on a countable set and distinguishes the different supports $ \supp(C_j) $), such that each $ C_j $ has a bounded support and $ \bigcup_j \supp(C_j) $ equals the region denoted in red in figure \ref{covering}. Then let us replace in equation \eqref{27} $ C $ with $ \sum_j \lambda_j(a_j s) C_j $

\begin{align} 
 & \| \int_0^t ds T_I(t,s) [T^*_E(s) \sum_j \lambda_j(a_j s) C_j T_E(s), B ] T^*_I(t,s)  \|   \\
 & \le \sum_j  \frac{|\L_j(a_jt)|}{a} \| [T^*_E(t) C_j T_E(t), B ] \| + \int_0^t ds \sum_j \frac{|\L_j(a_js)|}{a}  \Big( 2 \| I(s) \| \| [T^*_E(s) C_j T_E(s), B ] \| + \| \bigskip[T^*_E(s) [E(s), C_j] T_E(s), B \big] \|  \Big)  \label{147}
\end{align}
It is immediate to understand that the convergence of the sum in \eqref{147} is guaranteed by the locality of $ E(s) $ and by the exponential decrease of $ e^{-\frac{\dist(\supp \, B,\supp \, C_j)}{\xi}} $ and $ e^{-\frac{\dist(\supp \, B,\supp [E(s),C_j])}{\xi}} $. The overall decay with distance will be in relation to the closest perturbation to $ B $, therefore of $ O(e^{-\frac{l}{2\xi}}) $.

Another possible setting of ``sparse'' perturbations is given in figure \ref{fig_left_right}. The two regions of perturbations left and right, respectively in red and blue in the picture, are taken into account with a {\it two steps procedure}. In the first step we evaluate the effect of the region of perturbations left on the L-R bound of the unperturbed dynamics $ E $ of theorem \ref{Floquet_Perturbation}. In the second step we consider the dynamics resulting from the first step as the unperturbed one and consider the effect of the perturbations given by region right. The overall decay in distance will be of $ O(e^{-\frac{l}{\xi}}) $.


\begin{figure}[h]
\setlength{\unitlength}{1mm} 

\begin{picture}(180,45)(-30,-40)

\thicklines

\put(-15,-24){\line(0,1){4}}
\put(85,-24){\line(0,1){4}}
\put(-15,-22){\line(1,0){100}}

\linethickness{1mm}

\put(12,-23.5){\line(-1,0){5}}

{\color{yellow} \put(-15,-18){\line(1,0){27}}}
{\color{yellow} \put(62,-18){\line(1,0){22}}}

{\color{blue} \put(-17,-22){\line(1,0){100}}}
\put(65,-23.5){\line(-1,0){5}}

\thinlines

\put(2,-27){$ \supp(A) $}
\put(58,-27){$ \supp(B) $}
\put(10,-29){$\underbrace{\hspace{50mm}}$}
\put(35,-34){$ l$}

\put(92,-23){$ \Longleftarrow \hspace{2mm}  \textrm{unperturbed Hamiltonian}$}

\put(92,-18){$ \Longleftarrow \hspace{2mm} \textrm{perturbations} \, \{ C_j \} $}

\end{picture}
\caption{A set of perturbations $ \{ C_j \} $ that cover the support of $ A $ and extends from the left end of the system to roughly half the way in between $ A $ and $ B $. } 
\label{fig_left_right}
\end{figure}
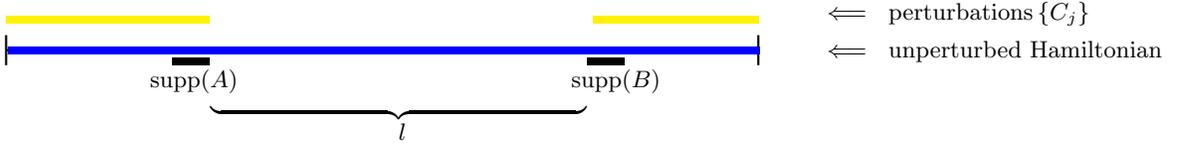


\section{Application I: the perturbed Anderson model} \label{applications}

\subsection{The case of a single perturbation}

We study the one-dimensional $ XY $ model with random magnetic field perturbed by a sparse set of $ ZZ $ terms with coupling constant $ \Delta $. For simplicity let us consider for now the case where there is a single $ ZZ $ term located on the sites $ 0 $ and $ 1 $, later on we will generalize our result to the case of a sparse set of $ ZZ $ terms. The system's Hamiltonian reads:
\begin{align} \label{Ham}
 H_\omega = \sum_{j=-L}^{L-1} J \left( \sigma_j^x\sigma_{j+1}^x + \sigma_j^y\sigma_{j+1}^y \right) + \sum_{j=-L}^{L} \omega_j \sigma_j^z  + \Delta \sigma_0^z\sigma_1^z  
\end{align}
The Hilbert space is $ \bigotimes_{j=-L}^L \mathds{C}^2 $, and $ \sigma^x $, $ \sigma^y $, $ \sigma^z $ denote the Pauli matrices. $ \sigma_j^x\sigma_{j+1}^x $ is a short hand for $ \mathds{1}_2 \otimes \dotsm \otimes \sigma_j^x \otimes \sigma_{j+1}^x \otimes \dotsm \otimes \mathds{1}_2 $, and similarly for the other terms in \eqref{Ham}.
$ \omega_j $ is taken at random and uniformly, from the interval $ [-\Omega,\Omega] $, and $ \omega = \{\omega_j\}_{j\in \{-L,...,L\}} $.

The model given in \eqref{Ham} with $ \Delta = 0 $ has been shown to be equivalent to the Anderson model and its dynamics is such that the supports of operators do not spread, see section 4 of \cite{Sims_Stolz_2012}, leading, with $ l:= \dist(\supp A,\supp B) $, to
\begin{align} \label{Anderson}
\underset{\omega}{\mathbb{E}} \sup_{t} \b| [e^{itH_{\omega,\Delta=0}} \, A \,e^{-itH_{\omega,\Delta=0}},B] \b| \le    K \|A\|  \|B\|   \, e^{-\frac{l}{\xi} }
\end{align}
In \eqref{Anderson} $ A $ and $ B $ are two operators with fixed and bounded supports that are assumed, for simplicity, to be connected, and do not change upon rescaling the system's size.

From \eqref{Anderson} we also see that: 
\begin{align} \label{without_sup} 
 \b| [e^{itH_{\omega,\Delta=0}} \, A \,e^{-itH_{\omega,\Delta=0}},B] \b| \le \sup_{t} \b| [e^{itH_{\omega,\Delta=0}} \, A \,e^{-itH_{\omega,\Delta=0}},B] \b|
\end{align}
then taking $ \underset{\omega}{\mathbb{E}} $ both sides of \eqref{without_sup}:
\begin{align}
\underset{\omega}{\mathbb{E}}  \b| [e^{itH_{\omega,\Delta=0}} \, A \,e^{-itH_{\omega,\Delta=0}},B] \b|  \le \underset{\omega}{\mathbb{E}} \sup_{t} \b| [e^{itH_{\omega,\Delta=0}} \, A \,e^{-itH_{\omega,\Delta=0}},B] \b| \le    K \|A\|  \|B\|   \, e^{-\frac{l}{\xi} }
\end{align}


\textbf{Lemma} (Logarithmic lightcone in the Anderson model with a single $ ZZ $ perturbation, informal).
{\it Given the Hamiltonian $ H_\omega $ in \eqref{Ham} and $ A $ and $ B $ with fixed and bounded supports at a distance equal to $ l $, it holds:
\begin{align} \label{L-R}
\underset{\omega}{\mathbb{E}}   \b| [e^{itH_{\omega}} \, A \,e^{-itH_{\omega}},B] \b| \le    K \|A\|  \|B\|   \, e^{-\frac{l}{\xi} } + K' \, \|A\|  \|B\|  \, t \, e^{-\frac{d}{\xi} }
\end{align}
with $ d:= \max \{\dist(\supp A,[-2,3]),\dist(\supp B,[-2,3])\} $, and 
\begin{align} \label{delta_values}
 K' \propto
\begin{cases}
\Delta & \textrm{with} \hspace{2mm} 0 \le \Delta \le J \\
\frac{J(J+\Omega)}{\Delta} & \textrm{with} \hspace{2mm} \Delta \ge J + \Omega
\end{cases}
\end{align}
The constants $ K $ and $ \xi $ are the same as in equation \eqref{Anderson}. The minimum possible value of $ d $ is $ d \approx \frac{l}{2} $.
}

The result given in \eqref{L-R},\eqref{delta_values} shows that both very small and very large values of $ \Delta $ suppress the amplitude of the logarithmic lightcone.

The name ``logarithmic lightcone'' comes from the equality $ K' \, t \, e^{ -\frac{d}{\xi} } =   e^{\log (K' t) - \frac{d}{\xi} } $. This means that up to a time $ t $ exponentially large in $ d $,  $ e^{iH_\om t} A e^{-iH_\om t} $ and $ B $ commute up to an exponentially small correction. In fact defining $ t_{max} $ such that $ K' t_{max} = e^{\frac{d}{2\xi}} $, at $ t=t_{max} $ the upper bound in \eqref{L-R} is proportional to $ e^{-\frac{d}{2\xi}} \ll 1 $, when $ d \gg \xi $. $ e^{iH_\om t} A e^{-iH_\om t} $ is a propagating operator and if we imagine $ B $ as a ``detector'' of such propagation, according to \eqref{L-R} within the time $ t_{max} $ there is virtually no detection. In a logarithmic lightcone $ t_{max} $ scales as $ \frac{1}{K'} e^{\frac{d}{2\xi}} $.

We have already studied in \cite{Toniolo_2024_1} how a local perturbation of the Hamiltonian affects the Lieb-Robinson bounds of the dynamics. In particular the case $ \Delta \le J $ in \eqref{delta_values} follows directly from theorem 1 of \cite{Toniolo_2024_1}. We remark that in theorem 1 and lemma 2 of our work \cite{Toniolo_2024_1} we considered the perturbation of a completely generic local Hamiltonian, by a general set of local perturbations $ h_j $, the only restrictions on such perturbations being the rate of the exponential increase of their norm and the location of their supports with respect to the supports of $ A $ and $ B $. Here instead we solve in lemma \ref{single_ZZ} a given model \eqref{Ham}, so we take direct advantage of its explicit structure to study the regime $ \Delta \ge J + \Omega $. The physical interpretation of the L-R bound \eqref{L-R}, with $ \Delta \gg J + \Omega $, is immediate when the Hamiltonian \eqref{Ham} is written in the particle language after a Jordan-Wigner transformation \cite{Bosonization,Kitaev_2009}

\begin{align} \label{Ham_particle}
 H_\omega = \sum_{j=-L}^{L-1} J \left( \psi_{j+1}^\dagger \psi_{j} + \textrm{h.c.} \right) + \sum_{j=-L}^{L} \omega_j \psi_j^\dagger \psi_{j} + \Delta  \psi_0^\dagger  \psi_0  \psi_1^\dagger \psi_1
\end{align}
The Hamiltonian \eqref{Ham_particle} describes particles hopping on a one dimensional lattice, these interact among each other only on the sites $ 0 $ and $ 1 $ when both of them are occupied, otherwise outside of $ \{0,1\} $ the random potential $ \omega_j $, even for an infinitesimally small $ \Omega $, localize them. With $ \Delta \gg J + \Omega $ the configuration where both the lattice sites $ 0 $ and $ 1 $ are occupied becomes energetically unfavorable. Eventually, after a long time, the particles have a chance to hop from the left to the right (or vice versa) of the energy barrier in $ [0,1] $, as pairs, unless $ \Delta \rightarrow \infty $ in such a case the space of configurations is split into two disconnected regions. Indeed in such a limit equation \eqref{L-R} shows that Anderson dynamics is restored.

In the following lemma we give the exact statement and the proof of the bound \eqref{L-R}.


\begin{Lemma} \label{single_ZZ}
{ \it Given the Hamiltonian 
\begin{align} 
 H_\omega = \sum_{j=-L}^{L-1} J \left( \sigma_j^x\sigma_{j+1}^x + \sigma_j^y\sigma_{j+1}^y \right) + \sum_{j=-L}^{L} \omega_j \sigma_j^z  + \Delta \sigma_0^z\sigma_1^z  
\end{align}
as in \eqref{Ham}, $ A $ and $ B $ operators with a fixed bounded support, such that al least one of them has support not overlapping with $ [0,1] $, with $ d:= \max \{\dist(\supp A,[-2,3]),\dist(\supp B,[-2,3])\} $, $ K $ and $ \xi $ as in equation \eqref{Anderson}, the following L-R bounds, depending on the value $\Delta $, hold:
\begin{align} \label{first_l}
\underset{\omega}{\mathbb{E}}   \b| [e^{itH_{\omega}} \, A \,e^{-itH_{\omega}},B] \b| \le    K \|A\|  \|B\|   \, e^{-\frac{l}{\xi} } + 2K \|A\|  \|B\|   \, \Delta \, t \, e^{-\frac{d}{\xi} }
\end{align}
Also
\begin{align} \label{final_large_Delta}
\underset{\omega}{\mathbb{E}}   \b| [e^{itH_{\omega}} \, A \,e^{-itH_{\omega}},B] \b| \le & 2K \|A\|  \|B\|   \, e^{-\frac{l}{\xi} } + K \, \|A\|  \|B\|  \, \frac{J}{\Delta} \, \left( 16 + 64 \Omega t + 160 J t  \right) \, e^{-\frac{d}{\xi} } 
\end{align}
}
\end{Lemma}


\begin{proof}
The proof follows as an application of theorem \ref{Floquet_Perturbation}. We evaluate first the two bounds given in \eqref{direct} and \eqref{by_part} and then establish which is the better bound in the different regimes of $ \frac{\Delta}{J} $.

The bound given in \eqref{direct} is applied to the Hamiltonian \eqref{Ham}, with $ \l(t) = 1 $, $ \|C\|=\Delta $ and $ f(t)=K $, giving:
\begin{align} 
\underset{\omega}{\mathbb{E}}   \b| [e^{itH_{\omega}} \, A \,e^{-itH_{\omega}},B] \b| \le  K \|A\| \|B\|   \, e^{-\frac{l}{\xi} } + 2 K \Delta  \|A\|  \|B\|  \, t \, e^{-\frac{d}{\xi} }
\end{align}
with $ d := \max \{ \dist(\supp A, [0,1]), \dist(\supp B, [0.1]) \} $. The bound in \eqref{first_l} is a further upper bound with the value of $ d $ given in the statement of the lemma.

To apply the bound \eqref{by_part} of theorem \ref{Floquet_Perturbation} we will obtain $ E(t) $ after the application of the interaction picture to the dynamics of \eqref{Ham}, where the $XY$ part of the Hamiltonian \eqref{Ham}, that is the part proportional to $J$,  is unitarily transformed with respect to $ e^{-it \sum_{j=-L}^L \omega_j \sigma_l^z} e^{-it \Delta \sigma_0^z\sigma_1^z} $. The resulting dynamics will be given by an effective time-periodic Hamiltonian with frequency proportional to $ \Delta $.
We use the interaction picture to write:
\begin{align} 
& e^{-itH_\om }  = e^{-it \Delta \sigma_0^z\sigma_1^z} \, \mathcal{T} \left[ e^{ -i \int_0^t ds e^{is \Delta \sigma_0^z\sigma_1^z} \, \left( \sum_{j=-L}^{L-1} J \left( \sigma_j^x\sigma_{j+1}^x + \sigma_j^y\sigma_{j+1}^y \right) + \sum_{j=-L}^{L}  \omega_j \sigma_j^z \right) \, e^{-is \Delta \sigma_0^z\sigma_1^z}  } \right] \label{int_picture_0}  \\
& = e^{-it \Delta \sigma_0^z\sigma_1^z} \mathcal{T}  \left[ e^{ -i \int_0^t ds \left( e^{is \Delta \sigma_0^z\sigma_1^z}  \sum_{j=-1}^{1}  J \left( \sigma_j^x\sigma_{j+1}^x + \sigma_j^y\sigma_{j+1}^y \right)  e^{-is \Delta \sigma_0^z\sigma_1^z} +  \sum_{j=-L, j \neq \{-1,0,1\}}^{L-1}    J \left( \sigma_j^x\sigma_{j+1}^x + \sigma_j^y\sigma_{j+1}^y \right) + \sum_{j=-L}^{L} \omega_j \sigma_j^z \right)  } \right] \\
& = e^{-it \Delta \sigma_0^z\sigma_1^z} \mathcal{T}  \left[ e^{ -i \int_0^t ds \left( e^{is \Delta \sigma_0^z\sigma_1^z}  \sum_{j=-1}^{1}  2J \left( S_j^+ S_{j+1}^- + S_{j+1}^+ S_j^- \right)  e^{-is \Delta \sigma_0^z\sigma_1^z} +  \sum_{j=-L, j \neq \{-1,0,1\}}^{L-1} 2J \left( S_j^+ S_{j+1}^- + S_{j+1}^+ S_j^- \right) + \sum_{j=-L}^{L} \omega_j \sigma_j^z \right)  } \right]
\label{int_picture}
\end{align}
$ \mathcal{T} $ denotes time ordering, according to the rule: latest time goes on the left. 
In the last expression above we have introduced the operators $ S_j^+:= \frac{1}{2}(\sigma_j^x + i \sigma_j^y) $, and $ S_j^-=(S_j^+)^* $. These operators when supported at different sites commute.

We briefly mention that in our work \cite{Toniolo_2024_1} to consider perturbations of the dynamics of the many-body Anderson model, at the level of the Lieb-Robinson bounds, we used as the reference dynamics (namely the  dynamics that is factorized in the interaction picture) the one generated by the Anderson model itself, see equation (16) of \cite{Toniolo_2024_1}, here instead we are using as the reference dynamics the one given by the $ ZZ $ term.

Our fist goal is to evaluate the first term at the exponent inside the time-ordered operator. 
We immediately note the identity, with $ a \in \mathds{R} $:
\begin{align} \label{exp_sigma}
 e^{ia \sigma_0^z\sigma_1^z} := \sum_{l=0}^\infty \frac{1}{l!}(ia\sigma_0^z\sigma_1^z)^l = 
 \mathds{1}_4 \sum_{l=0}^\infty \frac{1}{(2l)!}(ia)^{2l} + \sigma_0^z\sigma_1^z \sum_{l=0}^\infty \frac{1}{(2l+1)!}(ia)^{2l+1} =  
 \mathds{1}_4 \cos a + i \sigma_0^z\sigma_1^z \sin a
\end{align}
Then, as a consequence of the fact that two distinct Pauli matrices anticommute, it follows 
\begin{align} \label{inv_trans}
 e^{is \Delta \sigma_0^z\sigma_1^z} S_0^+ S_1^-  e^{-is \Delta \sigma_0^z\sigma_1^z} = S_0^+ S_1^-
\end{align}

Using \eqref{exp_sigma}, it holds:
\begin{align} \label{74}
 e^{is \Delta \sigma_0^z\sigma_1^z} S_{-1}^- S_0^+ e^{-is \Delta \sigma_0^z\sigma_1^z} & = S_{-1}^- e^{is \Delta \sigma_0^z\sigma_1^z} S_0^+  e^{-is \Delta \sigma_0^z\sigma_1^z} \\
 & = S_{-1}^- [\mathds{1}_4 \cos (s\Delta) + i \sigma_0^z\sigma_1^z \sin (s\Delta)] S_0^+ [\mathds{1}_4 \cos (s\Delta) - i \sigma_0^z\sigma_1^z \sin (s\Delta)] \\
 & = S_{-1}^- \left( \cos^2 (s\Delta) S_0^+ + i\cos (s\Delta) \sin (s\Delta) [ \sigma_0^z\sigma_1^z,S_0^+] + \sin^2 (s\Delta) \sigma_0^z\sigma_1^z S_0^+  \sigma_0^z\sigma_1^z \right) \\
 & = S_{-1}^- \left( \cos (2s\Delta) + i \sin (2s\Delta) \sigma_1^z \right)  S_0^+ \\
 & = S_{-1}^- S_0^+ e^{i2s\Delta \sigma_1^z}   
 \label{15}
\end{align}
In the last step we have used $ [\sigma^z, S^+] = 2 S^+ $.
Similarly we get:
\begin{align} \label{75}
  e^{is \Delta \sigma_0^z\sigma_1^z} S_1^+ S_2^- e^{-is \Delta \sigma_0^z\sigma_1^z} & =  e^{i2s\Delta \sigma_0^z} S_1^+ S_2^- 
\end{align}
Overall we obtain that:
\begin{align}\label{u_trans}
 & e^{is \Delta \sigma_0^z\sigma_1^z} \, \sum_{j=-1}^{1}  \left( S_j^+ S_{j+1}^- + \textrm{h.c.} \right) e^{-is \Delta \sigma_0^z\sigma_1^z} = (S_{-1}^+  S_0^- e^{-i2s\Delta \sigma_1^z} + \textrm{h.c.}) + (S_0^+ S_1^- + \textrm{h.c.}) + (e^{i2s\Delta \sigma_0^z} S_1^+ S_2^- + \textrm{h.c.}) 
\end{align}

At this point we note that the interaction parameter $ \Delta $ only appears at the argument of $ \cos(2s\Delta) $ and $ \sin(2s\Delta) $ that are both zero-average periodic functions, then the idea is that the larger $ \Delta $ the ``smaller'' the influence on the dynamics of the interaction will be. To quantify what ``smaller''means we make use of theorem \ref{Floquet_Perturbation}. Let us first rewrite the full evolution operator. With
\begin{align} \label{def_E}
 E_\om := 2J \left( S_0^+ S_1^- + \textrm{h.c.} \right) + \sum_{j=-L, j \neq \{-1,0,1\}}^{L-1} 2J \left( S_j^+ S_{j+1}^- + \textrm{h.c.} \right) + \sum_{j=-L}^{L} \omega_j \sigma_j^z 
\end{align}
and with
\begin{align} \label{G_def}
 G_\om(t) := E_\om + 2J(S_{-1}^+ S_0^- e^{-i2s\Delta \sigma_1^z} + \textrm{h.c.}) + 2J(e^{i2s\Delta \sigma_0^z} S_1^+ S_2^- + \textrm{h.c.})
\end{align}
it is
\begin{align} 
& e^{-itH_\om }  = e^{-it \Delta \sigma_0^z\sigma_1^z} \, \mathcal{T} \left[ e^{ -i \int_0^t ds G_\om(s) } \right] 
\end{align}
The dynamics of the Hamiltonian $ E_\om $ in equation \eqref{def_E}, that is independent from $ \Delta $, gives rise to Anderson localization. In fact it holds:
\begin{align} \label{E_dyn}
\underset{\omega}{\mathbb{E}} \sup_t  \b| [e^{itE} \, A \,e^{-itE},B] \b| \le   2 K \|A\|  \|B\|   \, e^{-\frac{l}{\xi} }
\end{align}
with the constants $ K $ and $ \xi $ as in \eqref{Anderson}. We observe that $ E_\om = E_{1,\{\omega_0,\omega_1 \} } + E_{2,\{\om_j, j\in \{-L,-1\} \}} + E_{3,\{\om_j, j\in \{1,L\} \}} $, splits into three commuting pieces (we drop the disorder index $ \om $):
\begin{align}
& E_1 := 2J \left( S_0^+ S_1^- + \textrm{h.c.} \right) + \omega_0 \sigma_0^z + \omega_1 \sigma_1^z \label{E1} \\
& E_2 := \sum_{j=-L}^{-2} 2J \left( S_j^+ S_{j+1}^- + \textrm{h.c.} \right) + \sum_{j=-L}^{-1} \omega_j \sigma_j^z \label{E2} \\
& E_3 := \sum_{j=2}^{L-1} 2J \left( S_j^+ S_{j+1}^- + \textrm{h.c.} \right) + \sum_{j=2}^{L} \omega_j \sigma_j^z \label{E3} 
\end{align}
Then $ e^{itE} = e^{itE_1} e^{itE_2} e^{itE_3} $. To prove \eqref{E_dyn} we need to take into account all the possible relative positions of the supports of $ A $ and $ B $ and the intervals (or more precisely sets): $ [0,1] $, $ [-L,-1] $ and $ [2,L] $. Let us start considering the case where both $ \supp(A) $ and $ \supp(B) $ are inside the interval $ [-L,-1] $, then:
\begin{align} 
\underset{\omega}{\mathbb{E}}   \b| [e^{itE} \, A \,e^{-itE},B] \b| = \underset{\omega}{\mathbb{E}}   \b| [e^{itE_2} \, A \,e^{-itE_2},B] \b|  \le K \|A\|  \|B\|   \, e^{-\frac{l}{\xi} }
\end{align}
If $ \supp(A) $ is inside $ [-L,-1] $ and $ \supp(B) $ overlaps with $ [-L,-1] $ and $ [0,1] $, then: \begin{align} 
& \underset{\omega}{\mathbb{E}}   \b| [e^{itE} \, A \,e^{-itE},B] \b| = \underset{\omega}{\mathbb{E}} \b| [e^{itE_1} e^{itE_2} \, A \,e^{-itE_2} e^{-itE_1},B] \b|  
= \underset{\omega}{\mathbb{E}} \b| [ e^{itE_2} \, A \,e^{-itE_2} ,e^{-itE_1} \, B \, e^{itE_1} ] \b|  
\end{align} 
We define $ B'(t) := e^{-itE_1} \, B \,e^{itE_1} $. 
Denoting $ B_{r}(\supp A) $ the ball of radius $ r:= \dist(\supp A,\supp B) $ around the support of $ A $, and $ B_{r}(\supp A)^c $ its complement in $ [-L,L] $, and remarking that $ \dist(\supp(A),\supp(B)) = \dist(\supp(A),\supp(B'(t))) $, we have: 
\begin{align}
 & \underset{\omega}{\mathbb{E}}  \b| [e^{itE_2} \, A \, e^{-itE_2},B'(t)] \b| \\
 & = \underset{\omega}{\mathbb{E}}  \b| [ e^{itE_2} \, A \, e^{-itE_2} - \frac{1}{2^{|B_{r}(\supp A)^c|}} \Tr_{B_{r}(\supp A)^c} \left( e^{itE_2} \, A \, e^{-itE_2} \right) \otimes \mathds{1}_{B_{r}(\supp A)^c} , B'(t)] \b| \\
 & \le 2 K \|A\|  \|B\|   \, e^{-\frac{r}{\xi} }
\end{align}
All the other possible relative positions of $ \supp(A) $ and $ \supp(B) $ are treated similarly.

We can now apply theorem \ref{Floquet_Perturbation}
taking as the unperturbed Hamiltonian $ E $ the Hamiltonian $ E $ defined in \eqref{def_E} and $ G $ as in \eqref{G_def}. 
The statement in theorem \ref{Floquet_Perturbation} generalizes immediately to the case of a sum of perturbations $ \sum_j \l_j(a_j t) C_j $ with the $ C_j $ such that the union of their supports is a bounded set (namely it does not scale with the system's size). Then we set: 
\begin{align}
 &\l_1(a_1 s) := \cos (2\Delta s) \\
 & C_1 := 2J  \left( S_{-1}^+ S_0^- + S_{1}^+ S_2^- + \textrm{h.c.} \right) \label{C1} \\
 &\l_2(a_2 s):= \sin (2\Delta s) \\ 
 & C_2 := 2iJ \left( ( \sigma_{0}^z(S_{1}^+ S_2^- - \textrm{h.c.} ) -(S_{-1}^+ S_0^- - \textrm{h.c.})\sigma_{1}^z +  \right) \label{C2}
\end{align}
We upper bound $ \int_0^t ds \, |\cos(2\Delta s)| $ simply by $ t $. A more accurate estimate would lead to: 
\begin{align}
 \int_0^t ds \, |\cos(2\Delta s)| \le
\begin{cases} \label{t_delta}
 & \frac{\sin(2 \Delta t)}{2 \Delta}, \hspace{2mm} \textrm{with} \hspace{2mm} \ t \le \frac{\pi}{4 \Delta} \\
 & \frac{n}{2\Delta}, \hspace{2mm} \, \textrm{when} \hspace{2mm} \ \exists n > 1, n \in \mathds{N}:  (n-1) \frac{\pi}{4 \Delta} < t \le n \frac{\pi}{4 \Delta} 
\end{cases}
\end{align}
If the second condition in \eqref{t_delta} applies, it also holds:
$
 \int_0^t ds \, |\cos(2\Delta s)| \le \frac{n}{2\Delta} < \frac{2}{\pi}t + \frac{1}{2\Delta} 
$. 
We see that this bound would simply give, for large $ t $, an improvement of a factor $ \frac{\pi}{2} $ with respect to the trivial one.

We now apply \eqref{by_part} of theorem \ref{Floquet_Perturbation} to $ G_\om(t) $ of equation \eqref{G_def}. Going back to the Pauli matrices we get: 
\begin{align}
 \|[E,C_1]\| & \le J \| [ \sigma_{-1}^x\sigma_{0}^x  + \sigma_{-1}^y\sigma_{0}^y , E_1 + E_2 ] \| + J\| [ \sigma_{1}^x\sigma_{2}^x + \sigma_{1}^y\sigma_{2}^y, E_1 + E_3 ] \| \\
 & \le J \| [ \sigma_{-1}^x\sigma_{0}^x  + \sigma_{-1}^y\sigma_{0}^y \, , \, J \sigma_{0}^x\sigma_{1}^x + J \sigma_{0}^y\sigma_{1}^y +  J\sigma_{-2}^x\sigma_{-1}^x + J\sigma_{-2}^y\sigma_{-1}^y + \om_0 \sigma_0^z + \om_{-1} \sigma_{-1}^z ] \| \nonumber \\
 & + J\| [ \sigma_{1}^x\sigma_{2}^x + \sigma_{1}^y\sigma_{2}^y \, , \, J \sigma_{0}^x\sigma_{1}^x + J \sigma_{0}^y\sigma_{1}^y +  J\sigma_{2}^x\sigma_{3}^x + J\sigma_{2}^y\sigma_{3}^y + \om_1 \sigma_1^z + \om_2 \sigma_2^z] \| \\
 & \le 8 J(J+2\Omega) \label{num}
\end{align}
Equation \eqref{num} can be obtained from the algebra of Pauli matrices and can be checked numerically. The same holds for $ \|[E,C_2]\| \le 8 J(J+2\Omega) $. With $ \|C_1\| = \|C_2\| = 4J $, it is:
\begin{align}
 \tilde{f}_{C_1}(t) + \tilde{f}_{C_2}(t) \le 16KJ + 32KtJ(J+2\Omega) + 128KJ^2 t 
\end{align}
We then define $ d $ as in equation \eqref{d_def} 
\begin{align} 
 d:= \max \{\dist(\supp A,[-2,3]),\dist(\supp B,[-2,3])\}
\end{align}
Overall it is:
\begin{align} 
\underset{\omega}{\mathbb{E}}   \b| [e^{itH_{\omega}} \, A \,e^{-itH_{\omega}},B] \b| \le    K \|A\|  \|B\|   \, e^{-\frac{l}{\xi} } + K \, \|A\|  \|B\|  \, \frac{J}{\Delta} \, \left( 16 + 64 \Omega t + 160 J t  \right) \, e^{-\frac{d}{\xi} }
\end{align}


\end{proof}

The bound \eqref{final_large_Delta} holds for a generic position of the supports of the operators $ A $ and $ B $ with respect to the support $ \{0,1\} $ of the perturbation $ \Delta \sigma_0^z \sigma_1^z $. Assuming that the supports of $ A $ and $ B $ are respectively on the left and on the right of $ \{0,1\} $ a slightly better bound arises due to the decomposition of $ E $ in commuting operators $ E = E_1 + E_2 + E_3 $. This is obtained not applying directly the result of theorem \ref{Floquet_Perturbation} but rather repeating its proof, for the case of large $\Delta$, taking advantage, among other things of the fact that with that choice of the position of the supports, $ A $ and $ E_3 $ commute, and $ B $ and $ E_1 $ as well. 
This would improve the bound of roughly a factor $ \frac{1}{2} $

{\it Remark}. We like to compare our result \eqref{final_large_Delta} with a previous result given by \cite{Gebert_2022}. Their theorem 3.1 provides a L-R bound for the perturbed dynamics with an amplitude that is proportional to the inverse of the intensity of the perturbations, for perturbations that are supported among the operators $ A $ and $ B $. It appears that our approach, despite similar on one side, employing integration by part, is more general, in fact we do not make any assumption about the absence of degenerate eigenvalues in the spectrum of the perturbations, we allow a generic support (but bounded) of the perturbation $ C $, we also allow for the unperturbed dynamics to be completely generic, giving rise to a L-R bound with a time-dependence $ f(t,s) $ in \eqref{E_L-R}. Also when compared with lemma 4.4 of \cite{Gebert_2022} our overall constant appears to be of at least of two order of magnitude smaller. We remark that in the literature the dynamics of the Hamiltonian \eqref{Ham}, with the perturbation $ \Delta \sigma_{L-1}^z \sigma_{L}^z $ at right edge of the system, has been investigated by Huang in \cite{Huang_2023} finding that a $ \log t$-law, see Figure 1 of \cite{Huang_2023} describes the generation of entanglement entropy, measured by the von Neumann entropy, before the saturation of entanglement at (very) long times. This agrees with the fact that since such system displays a L-R bound with a logarithmic lightcone then the dynamically generated $\alpha$-Rényi entropy (with $0<\alpha \le 1$) increases at large time as $\log t $, as proven in \cite{Toniolo_2024_3}.


\subsection{Sparse perturbations of the Anderson model}

Let us consider a specific model for the perturbed systems depicted in figure \ref{supports}, that is the Anderson model with $ZZ$ terms extending on the left of the origin:
\begin{align} \label{Ham_sparse}
 H_\omega = \sum_{j=-L}^{L-1} J \left( \sigma_j^x\sigma_{j+1}^x + \sigma_j^y\sigma_{j+1}^y \right) + \sum_{j=-L}^{L} \omega_j \sigma_j^z  + \Delta \sum_{j=-L}^0 \sigma_j^z\sigma_{j+1}^z  
\end{align}
To apply theorem \ref{Floquet_Perturbation} and the discussion of section \ref{many_per} to this model we must start, as done in the proof the lemma \ref{single_ZZ}, with the transformation of the $XY$ terms of \eqref{Ham_sparse} through the unitary given by $\exp{-it \Delta \sum_{j=-L}^0 \sigma_j^z\sigma_{j+1}^z} $, namely we need to determine:
\begin{align} \label{full_int_pic}
 \exp{it \Delta \sum_{j=-L}^0 \sigma_j^z\sigma_{j+1}^z} \sum_{j=-L}^{L-1} J \left( \sigma_j^x\sigma_{j+1}^x + \sigma_j^y\sigma_{j+1}^y \right) \exp{-it \Delta \sum_{j=-L}^0 \sigma_j^z\sigma_{j+1}^z}
\end{align}
As already find out in the work of Fagotti, see equations 5.4 and 5.5 of \cite{Fagotti_2014}, the operator \eqref{full_int_pic} not only contain terms periodic in $\Delta t$ with zero average, as those arising in equations \eqref{74}, \eqref{75} of \ref{single_ZZ}, but also time-independent terms that are $\Delta$ independent and give rise to terms proportional to $ZXXZ$ and $ZYYZ$. Following the proof sketched in section \ref{many_per}, the periodic terms with zero average will give rise to a L-R bound with a logarithmic lightcone with a time dependence proportional to $ \frac{t}{\Delta} $, while the time-independent ones will give a logarithmic lightcone proportional to $t$ instead.


\subsection{Application II: perturbation of the XY model} \label{XY}

In this section we set to zero the random magnetic field in the Hamiltonian \eqref{Ham}, this corresponds to study the $XY$ model (that is free fermions) perturbed by $ \Delta \sigma_0^z, \sigma_1^z $. The application of theorem \ref{Floquet_Perturbation} implies again that in the case of a large local $ZZ$ perturbation the amplitude of the L-R bound is suppressed, becoming proportional to $ \frac{J}{\Delta} $. In the following we also consider, for simplicity, the setting where the operators $ A $ and $ B $ are supported respectively on the left and on the right of the perturbation. The L-R bound of the unperturbed Hamiltonian 
\begin{align}
 H_{XY} = \sum_{j=-L}^{L-1} J \left( \sigma_j^x\sigma_{j+1}^x + \sigma_j^y\sigma_{j+1}^y \right)
\end{align}
is given by
\begin{align} \label{L-R_XY}
\b| [e^{itH} \, A \,e^{-itH},B] \b| \le \|A\|  \|B\|  \, e^{v_{LR}t-l}
\end{align}
with $ l = \dist(\supp A, \supp B) $, and $ v_{LR} = 4eJ $ as given, for example, in \cite{Wang_2020,Toniolo_2024_2}, and the appendix \ref{induction} below. For this system the function $ f(t) $ of theorem \ref{Floquet_Perturbation} is given by $e^{v_{LR}t} $.

\begin{Corollary} \label{single_ZZ_no_dis}
{\it
Given the Hamiltonian 
\begin{align} \label{Ham_no_dis}
 H = \sum_{j=-L}^{L-1} J \left( \sigma_j^x\sigma_{j+1}^x + \sigma_j^y\sigma_{j+1}^y \right)   + \Delta \sigma_0^z\sigma_1^z  
\end{align}
defined on the Hilbert space $ \bigotimes_{j \in \{ -L,L \}} \mathds{C}^2 $, and two operators $ A $ and $ B $ such that  $ \supp A \subset [-L,-3] $, $ \supp B \subset [4,L] $, with $ l:= \dist(\supp(A),\supp(B)) $, and the assumption that $ \dist(\supp A,[0,1]) > \dist(\supp B,[0,1])\} $.  It holds:  
\begin{align} \label{L-R_dir}
\b| [e^{itH} \, A \,e^{-itH},B] \b| \le \left(1 + \frac{\Delta}{4eJ} \right) \, \|A\|  \|B\|  \, e^{v_{LR}t-d}
\end{align}
Also
\begin{align} \label{L-R_inv}
\b| [e^{itH} \, A \,e^{-itH},B] \b| \le   8 \frac{J}{\Delta} \|A\| \|B\|  e^{v_{LR}t-\dist(\supp A,\{-2\})} 
\end{align}
with $ v_{LR} := 4eJ $ the Lieb-Robinson velocity of the $ XY $ model.
}
\end{Corollary}

{\it Remark}. It is important to stress that in \eqref{L-R_inv} a large $ \Delta $ suppresses the amplitude of the L-R bound but does not change the linear light cone that is still determined by the value of $v_{LR}$. Instead in \eqref{L-R}, the factor $\Delta^{-1}$ amounts as a rescaling of time therefore changing the shape of the logarithmic lightcone.


\begin{proof}
The proof of \eqref{L-R_dir} follows directly from \eqref{direct} with $ \lambda = \Delta $, $ \|C\| = 1 $,
and $ \int_0^t ds e^{v_{LR} s} = \frac{1}{v_{LR}}(e^{v_{LR} t}-1) $.

Let us define $ F $ in analogy with \eqref{E1}-\eqref{E3}.
\begin{align} \label{def_F}
 F := J \left( \sigma_0^x\sigma_1^x + \sigma_0^y\sigma_1^y \right) + \sum_{j=-L, j \neq \{-1,0,1\}}^{L-1} J \left( \sigma_j^x\sigma_{j+1}^x + \sigma_j^y\sigma_{j+1}^y \right)  
\end{align}
Under the assumptions on the positions of the supports of $ A $ and $ B $, it holds $ [e^{itF}Ae^{-itF},B]=0 $, in fact $ F=F_1+F_2+F_3 $, with:
\begin{align}
& F_1 := J \left( \sigma_0^x\sigma_1^x + \sigma_0^y\sigma_1^y \right)  \\
& F_2 := \sum_{j=-L}^{-2} J \left( \sigma_j^x\sigma_{j+1}^x + \sigma_j^y\sigma_{j+1}^y \right) \\
& F_3 := \sum_{j=2}^{L-1} J \left( \sigma_j^x\sigma_{j+1}^x + \sigma_j^y\sigma_{j+1}^y \right) 
\end{align}
then $ e^{itF} = e^{itF_1} e^{itF_2} e^{itF_3} $. It follows that: 
\begin{align}
[e^{itF}Ae^{-itF},B] = [e^{itF_2}Ae^{-itF_2},B] =0
\end{align}
We define $ G(t) := F + \cos(2\Delta t) C_1 + \sin(2\Delta t) C_2 $
with $ C_1 $ and $ C_2 $ as given in equations \eqref{C1},\eqref{C2}.

From now on the proof is identical to that of lemma \ref{single_ZZ}, upon replacing $ E_1 $, $ E_2 $ and $ E_3 $ with $ F_1 $, $ F_2 $ and $ F_3 $. Therefore we simply report the final steps where we make explicit the integration in time of $ e^{v_{LR}s} $. We have that:
\begin{align}
 & \|[e^{itH} \, A \,e^{-itH},B]\|   \\
 & \le 4 \frac{J}{\Delta} \|A\| \|B\| e^{v_{LR}(t-s)-\dist(\supp A,\{-1\})} + 80 \frac{J^2}{\Delta} \|A\| \|B\|  \int_0^t ds e^{v_{LR}t-\dist(\supp A,\{-2\})}  \\
 & \le 4 \frac{J}{\Delta} \|A\| \|B\| e^{v_{LR}t-\dist(\supp A,\{-1\})} + 80 \frac{J^2}{\Delta v_{LR}} \|A\| \|B\|  e^{v_{LR}t-\dist(\supp A,\{-2\})} \\
 & \le 8 \frac{J}{\Delta} \|A\| \|B\|  e^{v_{LR}t-\dist(\supp A,\{-2\})} 
\end{align}

\end{proof}



%



%

%

\subsection*{Acknowledgments}
The authors  acknowledge support from UKRI grant EP/R029075/1. D. T. gladly acknowledges discussions on the topics of this work with Marius Lemm.

\section*{Appendix}
\appendix

\section{Proof by induction of the L-R bound for generic nearest neighbour time-dependent Hamiltonians} \label{induction}

In this appendix we present a remarkably simple proof of the L-R bound for generic nearest neighbour (NN) Hamiltonians of a one-dimensional spin system, in general time-dependent. We take a NN Hamiltonian just for simplicity, the result can be generalized to other strictly local Hamiltonians. The L-R bound follows from methods similar to those used in appendix B of \cite{Toniolo_2024_2}, and in turn introduced in \cite{Osborne_slides,Osborne_2007}. For a recent thorough study of L-R bounds see for example \cite{Wang_2020, Lucas_review}.

In analogy to the quantity $ \Delta_j(t) $ \cite{Osborne_slides,Osborne_2007} for time-independent Hamiltonians, we consider with $ T_{H_{\Lambda_{j}}}(t,s) := T\left(e^{-i \int_s^t du H_{\Lambda_{j}}(u) } \right) $ the time-ordered operator of unitary evolution generated by (in general) a time dependent Hamiltonian $ H_{\Lambda_{j}}(u) := \sum_{k=-j}^{j-1}H_{k,k+1})(u) $, the quantity:
\begin{align} \label{delta_t}
 \Delta_j(t,s) &:= \| T^*_{H_{\Lambda_{j+1}}}(t,s) A T_{H_{\Lambda_{j+1}}}(t,s) - T^*_{H_{\Lambda_{j}}}(t,s) A T_{H_{\Lambda_{j}}}(t,s) \| 
\end{align}

\begin{Lemma}
 Given $ \Delta_j(t,s) $ as defined in \eqref{delta_t} for a nearest neighbour Hamiltonian $ H $ of a one-dimensional spin system on the lattice $ \{-L,...,L\} $, $ L \in \mathds{N} $, and Hilbert space $ \mathcal{H}= \otimes_{j\in \{-L,...,L\} } \mathds{C}^d $, and $ A $ supported in $ [-1,1] $, with $ J := \max_{u \in (s,t)} \max_j \{\|H_{j,j+1}(u)\| \} $, it holds:
\begin{align}\label{delta_ind}
 \Delta_j(t,s) \le \|A\| \frac{[4J(t-s)]^j}{j!}
\end{align}

\end{Lemma}

\begin{proof}
 The proof is done by induction. Let us first consider, recalling, as in equation \eqref{54}, that for any Hermitian $ H $:
\begin{align} 
 T_H(t,0)=T_H(t,s) T_H(s,0) \Rightarrow T_H(t,s) = T_H(t,0) T^*_H(s,0) \Rightarrow i \frac{d}{ds} T_H(t,s) = - T_H(t,s) H(s)
\end{align}
It then holds:
\begin{align}
 \Delta_j(t,s) = & \| T^*_{H_{\Lambda_{j+1}}}(t,s) A T_{H_{\Lambda_{j+1}}}(t,s) - T^*_{H_{\Lambda_{j}}}(t,s) A T_{H_{\Lambda_{j}}}(t,s)  \| \label{A4}  \\ 
 & = \|  \int_s^t du \, \frac{d}{du} \left[ T^*_{H_{\Lambda_{j+1}}}(u,s) \left( T^*_{H_{\Lambda_{j}}}(t,u) AT_{H_{\Lambda_{j}}}(t,u) \right) T_{H_{\Lambda_{j+1}}}(u,s)  \right] \| \\
  & = \|  \int_s^t du \, i T^*_{H_{\Lambda_{j+1}}}(u,s)  \left[  H_{\Lambda_{j+1}}(u) - H_{\Lambda_{j}}(u) , T^*_{H_{\Lambda_{j}}}(t,u) A T_{H_{\Lambda_{j}}}(t,u) \right] T_{H_{\Lambda_{j+1}}}(u,s)   \| \label{A6}
\end{align}
$ H_{\Lambda_{j+1}}(u) - H_{\Lambda_{j}}(u) = H_{-j-1,-j}(u) + H_{j,j+1}(u) $ are the boundary terms of $ H_{\Lambda_{j+1}}(u) $. 
We start the proof by induction from $ j=1 $. From \eqref{A6} using trivial bounds it follows that:
\begin{align} \label{beginning}
 \Delta_1(t,s) \le \|A\| 4J(t-s)
\end{align}
\eqref{delta_ind} with $ j=1 $ gives \eqref{beginning}.
Let us now assume \eqref{delta_ind} and prove it for $ j+1 $. In analogy to \eqref{A4}-\eqref{A6} we have:
\begin{align}
 \Delta_{j+1}(t,s) = & \|  \int_s^t du \, i T^*_{H_{\Lambda_{j+2}}}(u,s)  \left[  H_{\Lambda_{j+2}}(u) - H_{\Lambda_{j+1}}(u) , T^*_{H_{\Lambda_{j+1}}}(t,u) A T_{H_{\Lambda_{j+1}}}(t,u) \right] T_{H_{\Lambda_{j+2}}}(u,s)   \| \label{A8}
\end{align}
We observe that $ H_{\Lambda_{j+2}}(u) - H_{\Lambda_{j+1}}(u) $, being supported on $ \{-j-2,-j-1,j+1,j+2\} $, commutes with $ T^*_{H_{\Lambda_{j}}}(t,s) A T_{H_{\Lambda_{j}}}(t,s) $. Then inserting this on the RHS of the commutator in \eqref{A8}, it is:
\begin{align}
 \Delta_{j+1}(t,s) =  \|  \int_s^t du \, i T^*_{H_{\Lambda_{j+2}}}(u,s)  \left[  H_{\Lambda_{j+2}}(u) - H_{\Lambda_{j+1}}(u) , T^*_{H_{\Lambda_{j+1}}}(t,u) A T_{H_{\Lambda_{j+1}}}(t,u) - T^*_{H_{\Lambda_{j}}}(t,s) A T_{H_{\Lambda_{j}}}(t,u) \right] T_{H_{\Lambda_{j+2}}}(u,s) \| \label{A9}
\end{align}
With trivial bounds and using \eqref{delta_ind} for the norm of the RHS of the commutator in \eqref{A8}, we get:
\begin{align}
 \Delta_{j+1}(t,s) \le \|A\| \frac{(4J)^{j+1}}{j!}  \int_s^t du (t-u)^j   \label{A10}
\end{align}
We conclude with a change of variable, $ v:=t-u $:
\begin{align}
  \int_s^t du (t-u)^j =  -\int_{t-s}^0 dv v^j = \frac{(t-s)^{j+1}}{j+1}
\end{align}

\end{proof}

\bibliography{bibliography_Single_ZZ_Up}

\end{document}